\documentclass[aps,
final,%
notitlepage,%
onecolumn,%
centertags,
floatfix, pra,
longbibliography]
{revtex4-1}
\usepackage{diagbox}
\usepackage{epstopdf}
\usepackage{stmaryrd}
\usepackage{epsfig}
\usepackage[notref,notcite]{showkeys}
\usepackage{mathtools}
\usepackage{graphicx}
\usepackage[normalem]{ulem}
\usepackage{cancel}
\usepackage{hyperref}
\usepackage{amsmath,amsthm,amssymb,amscd}
\usepackage{wasysym}
\usepackage{tikz}
\usepackage{multirow}
\usepackage{ulem}
\usepackage{dsfont}
\usepackage{array}

\usepackage{color}

\newcommand{\R}{\mathbb R}
\newcommand{\C}{\mathbb C}

\newcommand{\vertiii}[1]{{\left\vert\kern-0.25ex\left\vert\kern-0.25ex\left\vert #1 
    \right\vert\kern-0.25ex\right\vert\kern-0.25ex\right\vert}}

\newcommand{\Acal}{{\mathcal A}}
\newcommand{\Bcal}{{\mathcal B}}

\newcommand{\Hcal}{{\mathcal H}}

\newcommand{\PiAcal}{{\Pi_{\Acal}}}
\newcommand{\PiBcal}{{\Pi_{\Bcal}}}

\newcommand{\N}{\mathbb{N}}

\newcommand{\nab}{n_{\Acal,\Bcal}}
\newcommand{\nabmin}{\nab^{\min}}
\newcommand{\na}{n_{\Acal}}
\newcommand{\nb}{n_{\Bcal}}
\newcommand{\mab}{m_{\Acal,\Bcal}}
\newcommand{\Mab}{M_{\Acal, \Bcal}}
\newcommand{\KDNC}{{\mathcal N}_{\mathrm{NC}}}

\newtheorem{theorem}{Theorem}
\newtheorem{lemma}[theorem]{Lemma}
\newtheorem{definition}[theorem]{Definition}

\makeatletter
\newcommand{\oset}[3][0ex]{%
	\mathrel{\mathop{#3}\limits^{
			\vbox to#1{\kern-11\ex@
				\hbox{$\scriptstyle#2$}\vss}}}}
\makeatother

\begin{document}

\title{Kirkwood-Dirac nonclassicality, support uncertainty and complete incompatibility}

\author{Stephan De Bi\`evre$^{1}$}

\address{
$^1$Univ. Lille, CNRS, UMR 8524, Inria - Laboratoire Paul Painlev\'e, F-59000 Lille, France
}


\begin{abstract}
Given two orthonormal bases $\Acal=\{|a_i\rangle\}$ and $\Bcal=\{|b_j\rangle\}$ in a $d$-dimensional Hilbert space $\Hcal$,
 one associates to each state its Kirkwood-Dirac (KD) quasi-probability distribution. KD-nonclassical states -- for which the KD-distribution takes on negative and/or nonreal values -- have  been shown to provide a quantum advantage in quantum metrology and information, raising the question of their identification. 
Under  suitable conditions of incompatibility between the two bases, we provide sharp lower bounds on the support uncertainty of states that guarantee their KD-nonclassicality.  In particular, when the bases are \emph{completely incompatible}, a notion we introduce,
states whose support uncertainty is not equal to its minimal value $d+1$ are necessarily KD-nonclassical.
The implications of these general results for various commonly used bases, including the mutually unbiased ones, and their perturbations, are detailed. 

\end{abstract}
\pacs{vvv}

\maketitle
\section{Introduction}
The nonclassical features of quantum mechanical states can be of a very diverse nature. Incompatible, noncommuting and complementary observables, (de)coherence, interference, uncertainty principles,  negativity or non-reality of quasi-probability distributions, entanglement, noncontextuality and nonlocality constitute a non-exhaustive list of concepts used to evaluate the degree to which quantum states of a  variety of physical systems may or may not exhibit manifestly nonclassical behaviour in various experimental situations.  Partially in order to obtain a better understanding of quantum mechanics and partially because such nonclassical behaviour has proven essential for a number of tasks in quantum information theory and metrology, the study of their properties and of the relationship between them attracts continued attention~\cite{Jo07, Sp08,   HeWo10, Fe11, LuBa12, BaLu14, Pu14, Dr15, ThGiChHoBaLu16, HeEtAl16, YuSwDr18, Lo18, DeFaKa19, ArEtAl20, ArChHa20, CaHeTo20, UoEtAl21}. This is  in particular so when the Hilbert space of states $\Hcal$ is finite dimensional, as for systems of qudits or qubits, which is our focus here.

Given orthonormal bases $\Acal=\{|a_i\rangle\}$ and $\Bcal=\{|b_j\rangle\}$, the Kirkwood-Dirac (KD) distribution of a state $\psi$~\cite{Ki33, Di45} is the quasi-probability distribution
\begin{equation}\label{eq:KD}
Q(\psi)_{ij}=\langle a_i|\psi\rangle\langle\psi|b_j\rangle \langle b_j|a_i\rangle, 1\leq i,j\leq d,
\end{equation}
similar in spirit to the Wigner distribution~\cite{cagl69a, cagl69b} in continuous variable quantum mechanics. It is complex-valued and satisfies
$
\sum_{ij} Q(\psi)_{ij}=1,
$
with marginals
$
\sum_j Q(\psi)_{ij}=|\langle a_i|\psi\rangle|^2,  \sum_i Q(\psi)_{ij}=|\langle b_j|\psi\rangle|^2.
$
A state $\psi\in\Hcal$ is \emph{KD-classical} if its KD-distribution is real nonnegative everywhere so that its KD-distribution is a probability distribution. If not, it is \emph{KD-nonclassical}. 
KD-nonclassicality  is used  in quantum tomography~\cite{LuBa12, BaLu14, ThGiChHoBaLu16} as well as in the theory and applications of weak measurements, (non)contextuality, and their relation to nonclassical effects in quantum mechanics~\cite{Sp08, Dr15, Pu14, Lo18}. Also, KD-nonclassicality provides an operational quantum advantage in postselected metrology~\cite{ArEtAl20}.  

This poses the question  how to ensure the prevalence of KD-nonclassical states? The KD-distribution and hence the KD-nonclassicality of $\psi$  depend not only on $\psi$, but also on $\Acal$ and $\Bcal$. 
In applications their choice is linked to the identification of two observables $A$ and $B$ of which they are eigenbases. The questions are therefore: under what conditions on $\Acal$ and $\Bcal$ are most states KD-nonclassical and where in $\Hcal$ are those states located?

A form of incompatibility or of noncommutativity is needed between the projectors $|a_i\rangle\langle a_i|$ and $|b_j\rangle\langle b_j|$ for KD-nonclassical states to exist. Indeed, if for example $A$ and $B$ have nondegenerate spectra, and if they are compatible in the usual sense that  $[A,B]=0$, then all those projectors commute, and therefore each $|a_i\rangle$ is up to a phase equal to some $|b_j\rangle$. It is then immediate from~\eqref{eq:KD} that $Q_{ij}(\psi)\geq 0$ for all $\psi$ so that there are no KD-nonclassical states in $\Hcal$.  Under the assumption that $A$ and $B$ are incompatible in the sense that they do not commute, a sufficient but non-optimal condition for a state to be KD-nonclassical was given in~\cite{ArChHa20}.  In the present paper, we sharpen and optimize that result in several ways. We start our analysis from
the observation that the above notion of incompatibility, and others close to it~\cite{HeWo10, HeEtAl16,  DeFaKa19,  CaHeTo20, UoEtAl21}, is weak, since it is based on the negation of a very strong notion of compatibility, namely that $A$ and $B$ commute. We introduce two increasingly restrictive notions of incompatibility, \emph{strong incompatibility} and \emph{complete incompatibility}, not obtained by negating a form of  commutativity, but  emanating directly from an analysis of incompatible successive measurements.  

We show that under these hypotheses, the KD-(non)classical states can be identified through their \emph{support uncertainty}, a measure of their uncertainty with respect to their $\Acal$- and $\Bcal$-representations [See Eq.~\eqref{eq:suppunc}].  We illustrate the power of these general results by analyzing the KD-nonclassicality for  spin bases, for mutually unbiased bases (MUBs), including the discrete Fourier transform (DFT), and their perturbations.

\section{Strongly incompatible bases}\label{s:stroinc}
Let $U$ be the unitary transition operator between $\Acal$ and $\Bcal$, defined as 
$U|a_j\rangle=|b_j\rangle$, with matrix elements $U_{ij}=\langle a_i|b_j\rangle$ in $\Acal$. Let $\mab^2=\min_{i,j} |\langle a_i|b_j\rangle|^2$, $\Mab^2=\max_{i,j} |\langle a_i|b_j\rangle|^2$, then
$$
0\leq\mab\leq d^{-1/2}\leq \Mab \leq 1.
$$
It is easily seen that $\Acal$ and $\Bcal$ are incompatible in the sense of ``noncommuting'' iff there exist $1\leq i,j\leq d$ so that $0<|\langle a_i|b_j\rangle|^2<1$. 
Hence, if $\Mab<1$, the bases are incompatible. The converse is not true; however, we note that  $\Mab=1$ if and only if there exist $k\geq1$ states in $\Acal$ that are equal, up to a phase, to states in $\Bcal$. Reordering the bases, we have $|a_i\rangle=|b_i\rangle$ for $i\in\llbracket 1,k\rrbracket$.  The KD-distribution then satisfies $Q_{ij}=|\langle a_i|\psi\rangle|^2\delta_{i,j}$ for $i,j\in\llbracket 1,k\rrbracket$.  It is moreover block-diagonal and one can therefore  restrict the KD-nonclassicality analysis to the $(d-k)$-dimensional Hilbert subspace orthogonal to the common basis vectors. Without loss of generality, we therefore assume that $\Mab<1$. 
\begin{definition}\label{def:STROINC}
We say
$\Acal$ and $\Bcal$ are  \emph{strongly incompatible  (STROINC) bases}  if, for all $1\leq i,j\leq d$, $\langle a_i|b_j\rangle\not=0$ or, equivalently, if
$
0<\mab^2\leq |\langle a_i|b_j\rangle|^2\leq \Mab^2<1.
$
\end{definition}
Strong incompatibility is easily checked to be equivalent to the requirement that \emph{none} of the projectors $|a_i\rangle\langle a_i|$ commutes with \emph{any} of the $|b_j\rangle\langle b_j|$. It is therefore stronger than the usual notion of incompatibility, which only requires that at least one pair of those projectors does not commute.   For example, consider an integer spin $s$. Then it is easily checked using the explicit expression for the Wigner rotation matrices~\cite{SA94} that 
$$
m_{J_xJ_z}:=\min_{m, m'}|\langle m_x=m|m_z=m'\rangle|=0,
$$ 
where $J_x, J_y, J_z$ are the spin operators and, for $m\in\llbracket -s, s\rrbracket$, $J_\ell|m_\ell= m\rangle=m|m_\ell= m\rangle$. (See Appendix~\ref{s:spins}.) For example, when $s=1$,   $\langle m_z=0|m_x=0\rangle=0$.  Hence, $J_x$ and $J_z$ do not commute, but they are not strongly incompatible. 

The condition $\mab>0$ means that, if the system is in the state $|a_i\rangle$, and a measurement in the $\Bcal$-basis is made, then any of the post-measurement states $|b_j\rangle$ occurs with a nonvanishing probability $|\langle a_i|b_j\rangle|^2\geq \mab^2$. The same holds with the roles of $\Acal$ and $\Bcal$ reversed. 
 Consequently, the larger $\mab$, the greater the uncertainty in the measurement outcomes. Hence a larger value of $\mab$ points to a stronger incompatibility.
The uncertainty is maximal when $\mab^2=d^{-1}$, in which case one easily sees $|\langle a_i|b_j\rangle|^2=d^{-1}$ for all $i,j$ so that all the aforementioned measurement outcomes are equally probable.  In view of this observation such bases can be considered maximally incompatible.  They are known as  mutually unbiased bases (MUBs) and  have found numerous applications in various quantum information protocols~\cite{Iv81, CeBoKaGi02, LuSuPaStBa11, LuBa12,  AgBoMiPa18, FaKa19}. We refer to~\cite{PlRoPe06, DuEnBeZy10}  for reviews on MUBs.
We can conclude that strong incompatibility  naturally interpolates between mere incompatibility and mutual unbiasedness with the parameter $0<\mab\leq d^{-1/2}$ providing a measure of the strength of the incompatibility.

\section{Support uncertainty}\label{s:supunc}
 To decide if $\psi$ is KD-nonclassical with respect to two STROINC bases, we use 
its support uncertainty  $n_{\Acal, \Bcal}(\psi)$, a notion of uncertainty that has proven useful in various contexts previously~\cite{DoSta89, MaOzPr04, Tao05, GhoJa11, WiWi21}:
\begin{equation}\label{eq:suppunc}
n_{\Acal,\Bcal}(\psi):=n_{\Acal}(\psi)+ n_{\Bcal}(\psi).
\end{equation}
Here $n_\Acal(\psi)$ (respectively $n_\Bcal(\psi)$) is the number of nonvanishing $\langle a_i|\psi\rangle$ (respectively $\langle b_j|\psi\rangle$). 
One should think of $n_\Acal(\psi)$ as the size of the support or the ``spread'' of the probability distributions $|\langle a_i|\psi\rangle|^2$ and $\langle b_j|\psi\rangle|^2$ of the state $\psi$ in the $\Acal$- and $\Bcal$-representations, which is one possible  measure of their uncertainty. Many other such measures exist, notably the entropic ones~\cite{MaUf88, Coles17}.
More precisely, we introduce the $\Acal$-support and $\Bcal$-support of $\psi$:
$
S_\psi=\{i\in\llbracket 1, d\rrbracket \mid \langle a_i|\psi\rangle\not=0\}$
and
$ 
T_\psi=\{j\in\llbracket 1, d\rrbracket \mid \langle b_j|\psi\rangle\not=0\}.
$
Then
\begin{equation}\label{eq:supportsizes}
n_\Acal(\psi)=|S_\psi|\not=0,\ n_\Bcal(\psi)=|T_\psi|\not=0;
\end{equation}
here  $|S|$ denotes the number of elements in $S$.
Clearly, $n_{\Acal,\Bcal}(\psi)$ is  a global  measure  of the uncertainty inherent in $\psi$, with respect to these two representations. 
We also introduce the \emph{minimal support uncertainty} $\nabmin$ of the bases $\Acal, \Bcal$ as
$
\nabmin=\min_{\psi\not=0} n_{\Acal,\Bcal}(\psi).
$
We call the collection of all points  $(n_\Acal, n_\Bcal)$ in the $n_\Acal$-$n_\Bcal$ plane for which there is a $\psi$ so that $n_\Acal(\psi)=n_\Acal$ and $n_\Bcal(\psi)=n_\Bcal$ the uncertainty diagram of the bases; see Fig.~\ref{fig:uncdiag}. To delimit the uncertainty diagram from below, we use an uncertainty principle originally shown for the Fourier transform on finite groups~\cite{DoSta89}, but which has much larger validity~\cite{GhoJa11, WiWi21}.  It reads
\begin{equation}\label{eq:uncprincab}
n_\Acal(\psi)n_\Bcal(\psi)\geq \Mab^{-2}.
\end{equation}
This follows immediately from
\begin{eqnarray}\label{eq:NCupperbound}
1&=&|\sum_{ij} Q_{ij}|\leq \sum_{ij} |Q_{ij}|
\leq \Mab\sum_{ij}|\langle a_i|\psi\rangle| |\langle b_j|\psi\rangle|\nonumber\\
&\leq& \Mab(\sum_{_{i\in S_\psi}}|\langle a_i|\psi\rangle|)(\sum_{j\in T_\psi}|\langle b_j|\psi\rangle|)\nonumber\\
&\leq& \Mab\sqrt{n_\Acal(\psi)n_\Bcal(\psi)}.
\end{eqnarray}
\begin{figure*}[t!]
\begin{center}
\includegraphics[height=2.4cm, keepaspectratio]{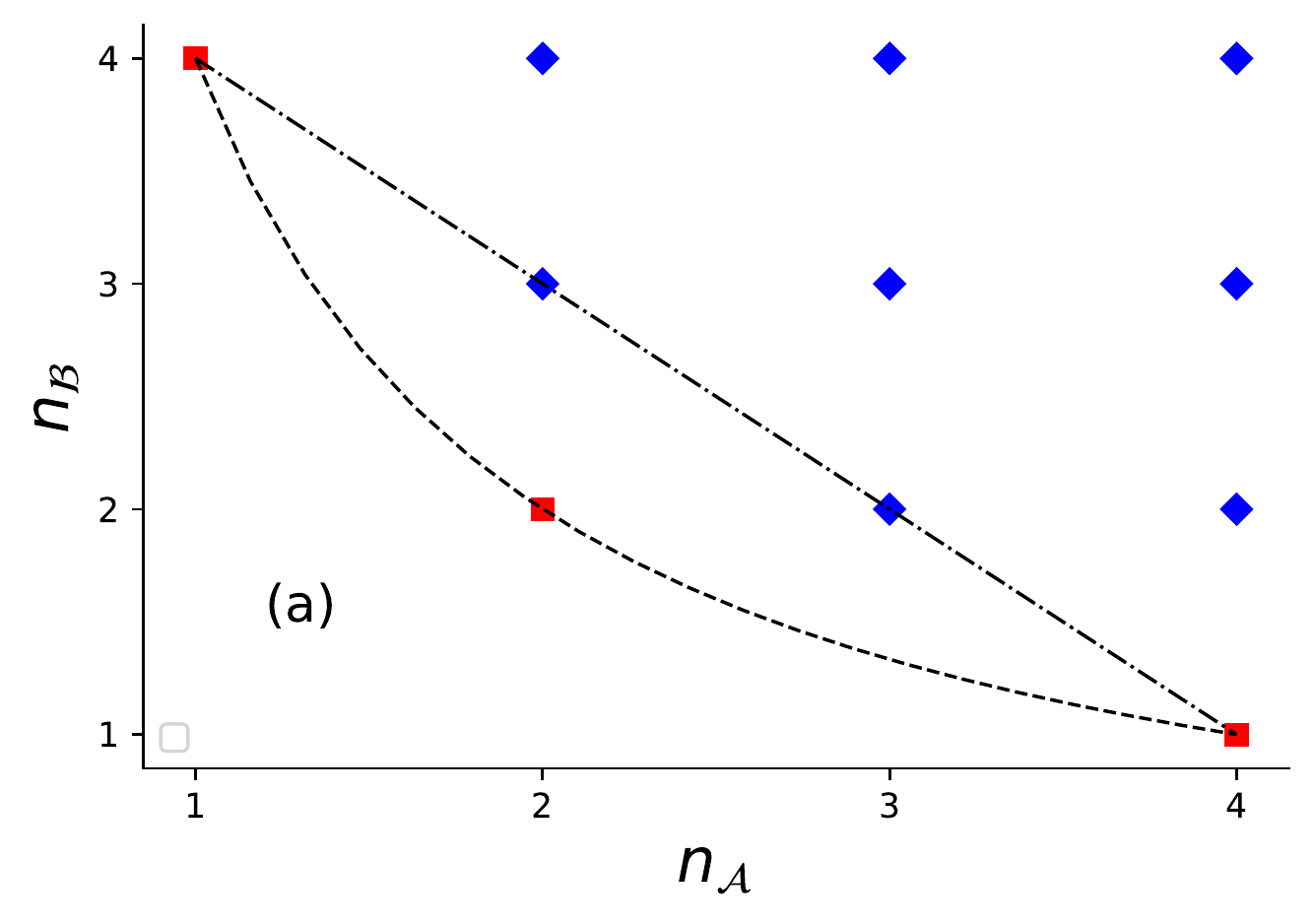}
\includegraphics[height=2.4cm, keepaspectratio]{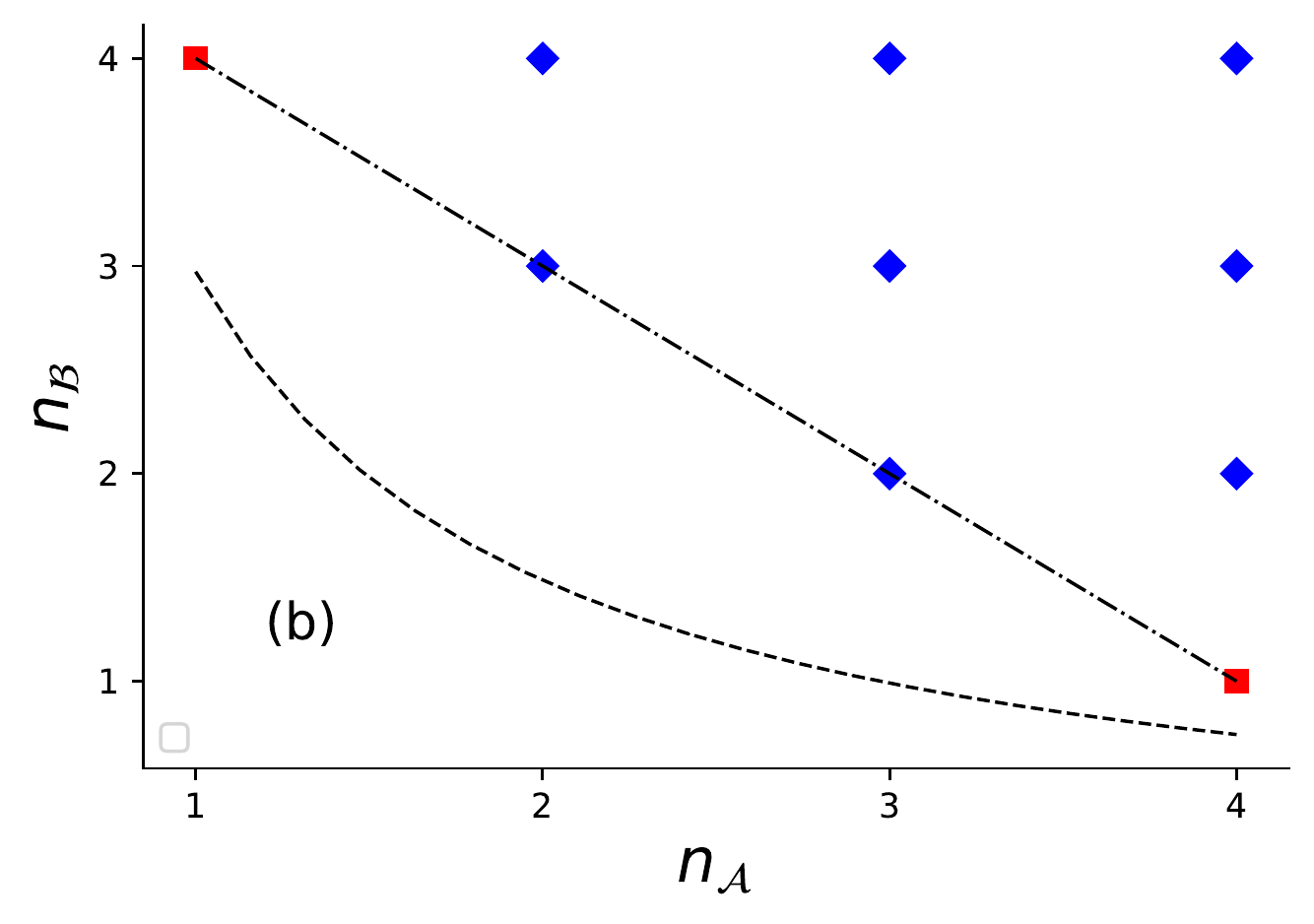}
\includegraphics[height=2.4cm, keepaspectratio]{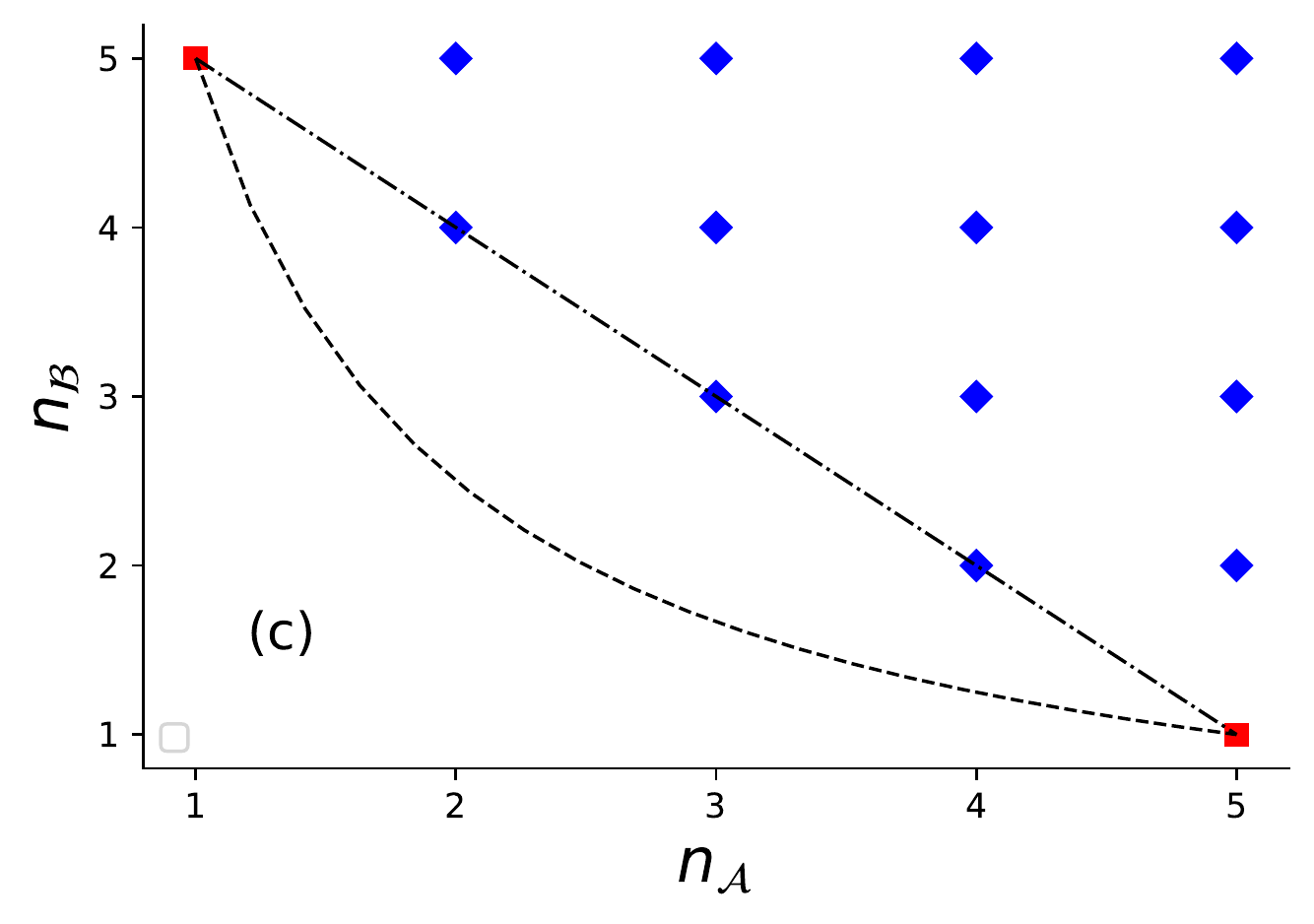}
\includegraphics[height=2.4cm, keepaspectratio]{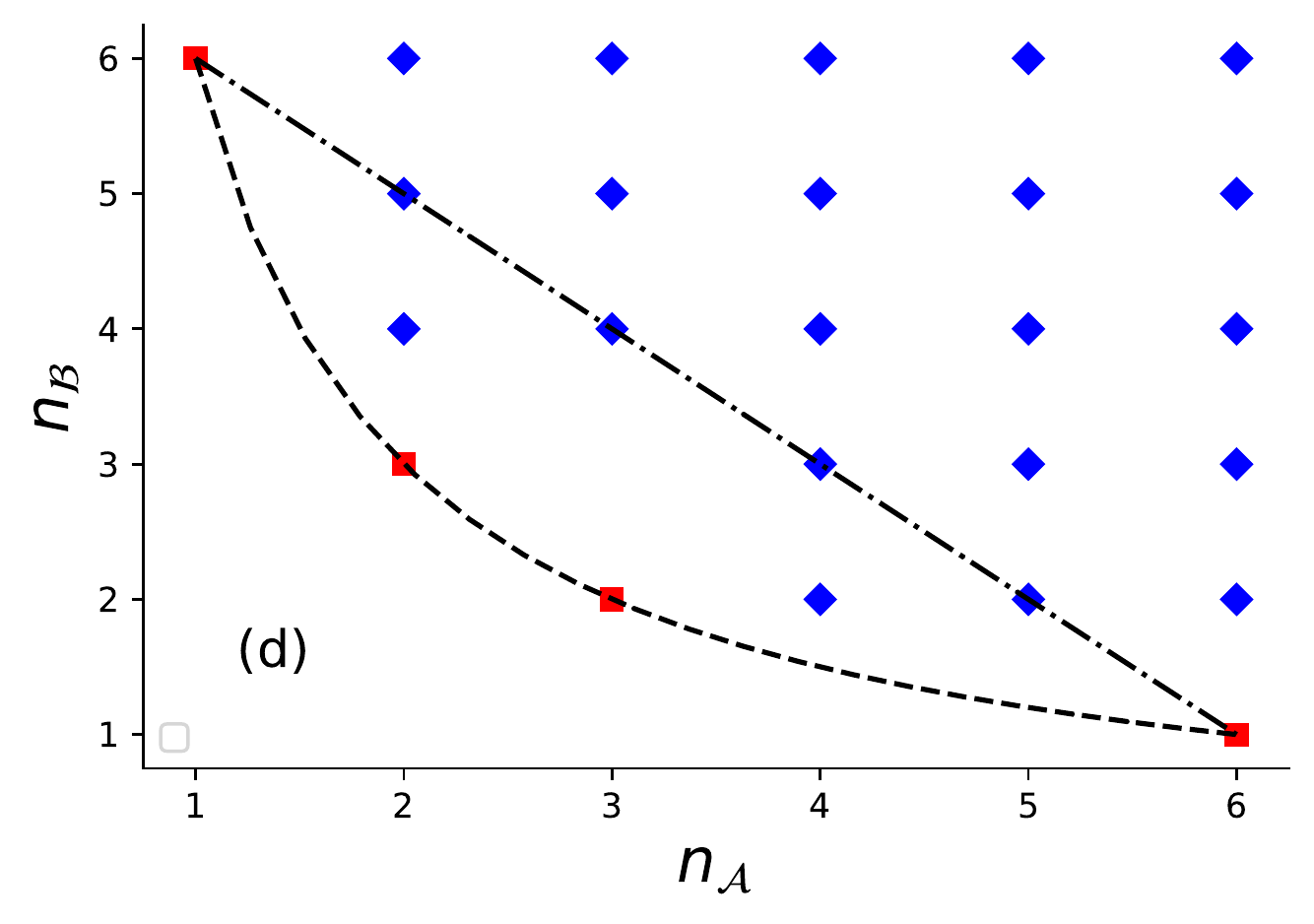}
\includegraphics[height=2.4cm, keepaspectratio]{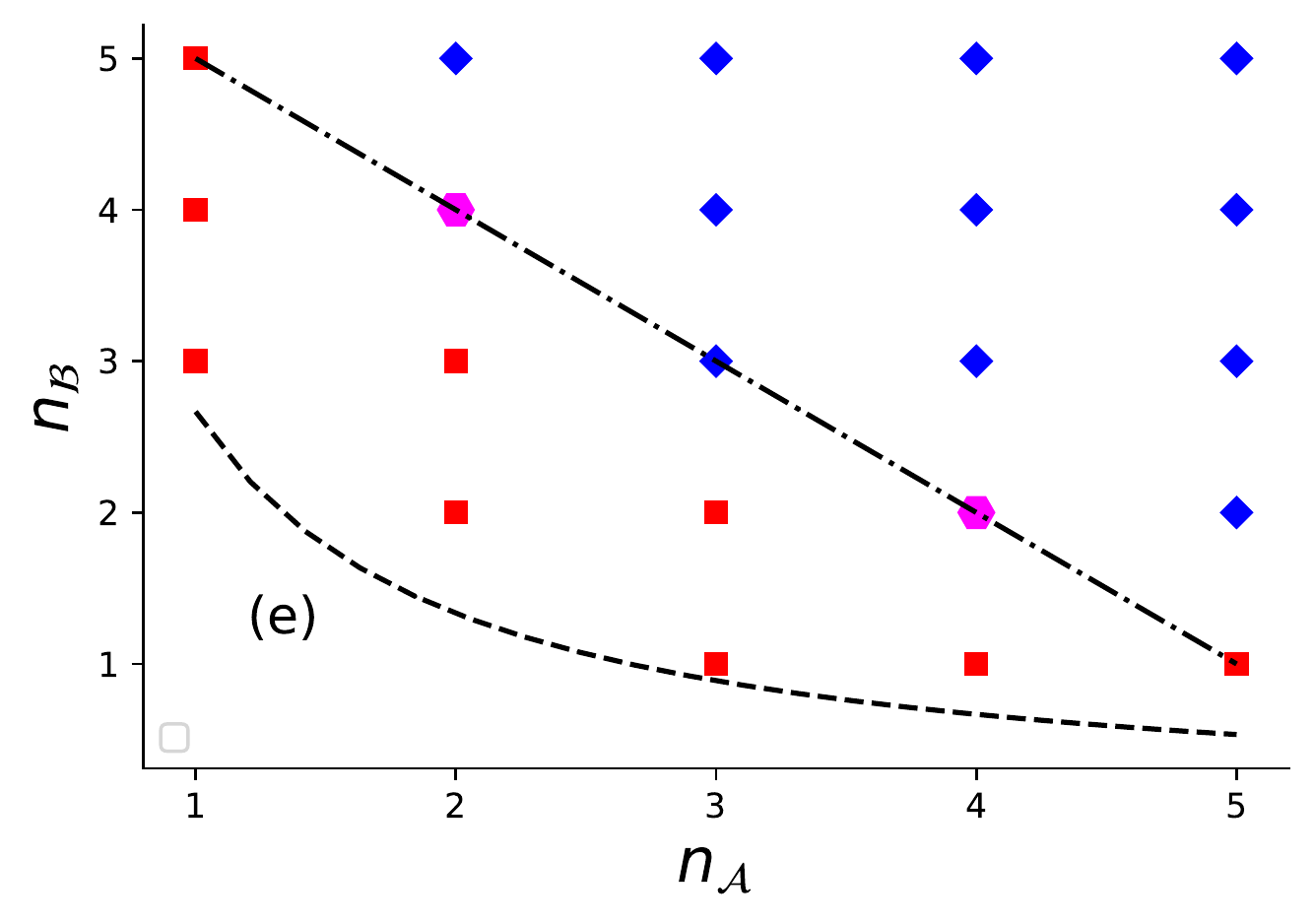}
\end{center}
\caption{Uncertainty diagrams.  Dashed curve: $n_\Acal(\psi)n_\Bcal(\psi)=\Mab^{-2}$. Dot-dashed line: $n_\Acal(\psi)+n_\Bcal(\psi)=d+1$. Diamonds (blue): KD-nonclassical states. Squares (red): KD-classical states.  
 (a) Complex MUBs; $d=4$, $\Mab^{-2}=4$, $\nabmin=4<5$. 
(b) Perturbed MUB  as in Eqn.~\eqref{eq:Upert}; $d=4$, $\epsilon=0.1$, $\Mab^{-2}=2.97<4$, $\nabmin=5$.   (c)  DFT; $d=5$, $\Mab^{-2}=5$, $\nabmin=6$.  (d) DFT, $d=6$, $\Mab^{-2}=6$, $\nabmin=5<7$. (e) Spin 2 transition matrix; $d=5$, $\Mab^{-2}=8/3<5$, $\nabmin=4<6$. Hexagons (magenta): KD-classical and KD-nonclassical states. \label{fig:uncdiag}} 
\end{figure*}

It follows the entire uncertainty diagram lies above or on the hyperbola $n_\Acal n_\Bcal=\Mab^{-2}$.  Eqn.~\eqref{eq:uncprincab} implies that
\begin{equation}\label{eq:uncprincaplusb}  
\nabmin\geq \frac{2}{\Mab}.
\end{equation}
Eqn.~\eqref{eq:uncprincab} yields more information when $\Mab$ is small than when $\Mab$ is large. In the case of MUBs, Eqns.~\eqref{eq:uncprincab} and~\eqref{eq:uncprincaplusb} become $n_\Acal(\psi)n_\Bcal(\psi)\geq d$, $\nabmin\geq 2\sqrt d$.   Fig.~\ref{fig:uncdiag} (a)-(d) shows representative examples of uncertainty diagrams of STROINC bases, on which the KD-(non)classical nature of the states has also been indicated. [For the details of the computations, see below and the Appendices.] In dimension $4$, all MUBs are known~\cite{Ba19}. They all display the same uncertainty diagram, shown in Fig.~\ref{fig:uncdiag}-(a).  Fig.~\ref{fig:uncdiag}-(b) shows the uncertainty diagram for a perturbed MUB matrix $U$ of the form 
\begin{equation}\label{eq:Upert}
U(\epsilon)=\exp(-i\epsilon L)U,
\end{equation}
where $L$ is self-adjoint. In the figure, $L_{jk}=-L_{kj}=i$, $1\leq j<k\leq d$.  Fig~\ref{fig:uncdiag}-(c)\&(d) concern bases with transition matrix  
$
U_{\textrm{DFT},i,j}:=\langle a_i|b_j\rangle=\frac1{\sqrt{d}}\exp(i\frac{2\pi}{d}{ij}),
$
the discrete Fourier transform (DFT) in dimensions $d=5$ and $d=6$.  Fig.~\ref{fig:uncdiag}-(e) shows the uncertainty diagram for the eigenbases of $J_z$ and $J_x$ for a spin $2$. As indicated above, these bases are not STROINC. One notes the presence of numerous classical states in the uncertainty diagram, a phenomenon we shall explain. 

Our results below show that the observed location of the KD-(non)classical states in these and other uncertainty diagrams can be predicted from the incompatibility properties of the bases used and the uncertainty properties of the states considered.

\section{Characterizing KD-nonclassicality}\label{s:charKDNC}
One observes in all panels of Fig.~\ref{fig:uncdiag} that there are no KD-classical states above the \emph{nonclassicality edge}, by which we mean the line segment in the first quadrant defined by $\na+\nb=d+1$. This is explained by the following theorem.
\begin{theorem}\label{thm:NCbound_tris} Let $\Acal, \Bcal$ be STROINCs on a $d$-dimensional Hilbert space $\Hcal$.  Then, if 
$\psi\in\Hcal$ satisfies
\begin{equation}\label{eq:supuncbd}
 n_{\Acal, \Bcal}(\psi)> d+1,
\end{equation}
then $\psi$ is KD-nonclassical. Equivalently,   if $\psi$ is KD-classical then $n_{\Acal, \Bcal}(\psi)\leq d+1$.
\end{theorem} 
The proof, which sharpens an argument in~\cite{ArChHa20}, is given in Appendix~\ref{s:proofthm}.
The lower bound in~\eqref{eq:supuncbd} is optimal in the sense that the support uncertainty of the basis vectors of STROINCs, which are all KD-classical, equals $d+1$. Nevertheless, while both optimal and sufficient for KD-nonclassicality, \eqref{eq:supuncbd} is not necessary: as can be observed in Fig.~\ref{fig:uncdiag}~(a)-(d),  KD-nonclassical states may occur below and on the nonclassicality edge.
While the transition matrix between the $J_z$ and $J_x$ bases for a spin $2$ is not STROINC since it contains zeroes, it will be shown  elsewhere~\cite{SDB21} that the conclusion of the theorem holds under weaker conditions than strong incompatibility that are satisfied in this case.  This explains the absence of KD-classical states above the nonclassicality edge in Fig.~\ref{fig:uncdiag}-(e).

In~\cite{ArChHa20}, it was shown that, if none of the $|a_i\rangle$ are equal (up to a phase) to any of the $|b_j\rangle$, then the condition
\begin{equation}\label{eq:NCsuff32}
n_\Acal(\psi)+n_\Bcal(\psi)> \lfloor{3d/2}\rfloor
\end{equation}
implies $\psi$ is KD-nonclassical; here $\lfloor x\rfloor$ is the integer part of $x$. 
 While this estimate holds under weaker conditions on the overlaps $\langle a_i|b_j\rangle$ than~$\mab>0$, it is not optimal as
 for example  when the bases are STROINC and if $d\geq 4$, since then  $\lfloor 3d/2\rfloor>d+1$; the difference between the lower bounds is increasingly pronounced for larger $d$. For example, for $d=6$ (see Fig.~\ref{fig:uncdiag}~(d)), our bound shows all states with $\nab\geq 8$ are KD-nonclassical, while the bound of~\cite{ArChHa20}  only guarantees this if $\nab\geq 10$.  

The theorem provides an upper bound to the support uncertainty of KD-classical states. As such, it is interesting to compare the situation with the one familiar from quantum mechanics and quantum optics where the ``classical'' states are the coherent states. They can be characterized as those that minimize the total noise $(\Delta Q)^2+(\Delta P)^2$, a known measure of uncertainty for pure states~\cite{hi89, yabithnaguki18, Bievre19, kwtakovoje19}. Here $Q$ and $P$ are two conjugate observables that satisfy the canonical commutation relation $[Q,P]=i$ which expresses the very  strong sense in which  they fail to commute, and -- in this sense -- their very strong incompatibility. The question then naturally arises in the discrete setting as well under which condition on the bases $\Acal$ and $\Bcal$  the support uncertainty of all  KD-classical states equals the minimal support uncertainty $\nabmin$? Note that this is the case in Fig.~\ref{fig:uncdiag}-(b)\&(c), but \emph{not} in Fig.~\ref{fig:uncdiag}-(a), (d)\&(e). 
We show below that a sufficient condition for this to happen is that the two bases are \emph{completely incompatible}, a notion we now introduce.

\section{Complete incompatibility} \label{s:coinc}
\subsection{Completely incompatible bases: definition}
Let $\Acal$ be as above and define, for any index set $S\subset \llbracket 1,d\rrbracket$, the orthogonal projector 
$$
\PiAcal(S)=\sum_{i\in S}|a_i\rangle\langle a_i|.
$$
We write $\Pi_\Acal(S)\Hcal$ for the $|S|$-dimensional subspace of $\Hcal$ onto which it projects. 
One should think of  $\Pi_\Acal(S)\Hcal$ as the set of all states $\psi$ whose $\Acal$-support $S_\psi$ lies in $S$.

  

\begin{definition}\label{def:TOINC}
We say that two bases $\Acal$ and $\Bcal$ are completely  incompatible (COINC) if and only if all index sets $S,T$ in $\llbracket 1,d\rrbracket$ for which $|S|+|T|\leq d$ have the property that $\PiAcal(S)\Hcal\cap\PiBcal(T)\Hcal=\{0\}$. 
\end{definition}

As a first motivation for this definition, we point out that, while it is only formulated in finite dimension, conjugate operators $Q$ and $P$ satisfy an analogous property. Indeed, it is well known that there do not exist states $\psi$ for which both the $Q$-representation $\psi(x)$ vanishes outside some bounded set $S\subset \R$ and the $P$-representation $\hat \psi(p)$ vanishes outside some bounded set $T\subset\R$~\cite{FoSi97}.
The above definition naturally transcribes this crucial property of conjugate operators to the finite-dimensional case. The restriction $|S|+|T|\leq d$ is then unavoidable since, for dimensional reasons, whenever $|S|+|T|>d$, the intersection $\PiAcal(S)\Hcal\cap\PiBcal(T)\Hcal$ must be nontrivial. 

While the definition is purely algebraic, its physical interpretation is readily given in terms of the quantum theory of selective projective measurements~\cite{Jo07, NiChu10, CTDL15}.  Given a basis $\Acal$, to every  partition $S_1,\dots, S_L$ of $\llbracket 1,d\rrbracket$ we associate  a projective partition of unity $\PiAcal(S_1), \dots\PiAcal(S_L)$. If initially the system is in the state $\psi\in\Hcal$ and the outcome $S_\ell$ is realized in a selective projective measurement then the post-measurement state is $\PiAcal(S_\ell)\psi/\|\PiAcal(S_\ell)\psi\|$.  The probability of this outcome is $\|\PiAcal(S_\ell)\psi\|^2$. This is referred to as a measurement in the basis $\Acal$.  When $L=d$ and $S_\ell=\{\ell\}$, for $\ell=1,\dots, d$,  the measurement is said to be fine-grained. Otherwise it is coarse-grained. When $L=2$, the corresponding measurement has only two possible outcomes and one has  $S_1=S, S_2=S^c$, for some $S\subset\llbracket 1,d\rrbracket$. 
Here $S^c=\llbracket 1,d\rrbracket\setminus S$, the complementary set of $S$. Note that projectors are observables with eigenvalues $1$ and $0$. Repeated  selective measurements of $\Pi_\Acal(S)$ and of $\Pi_\Bcal(T)$ on a system initially prepared in  $\psi$ systematically yield the outcome $1$ if and only if $\Pi_\Bcal(T)\Pi_\Acal(S)\psi$ belongs to $\Pi_\Acal(S)\Hcal\setminus\{0\}$.   Hence, if this occurs, $\PiAcal(S)\Hcal\cap \PiBcal(T)\Hcal\not=\{0\}$. In other words, the definition of COINC bases (COINCs) is equivalent to the statement that such repeated \emph{compatible selective measurements}  cannot occur for any $S,T$ for which $|S|+|T|\leq d$; they cannot occur for any insufficiently coarse-grained measurements. Since, whenever $|S|+|T|>d$, necessarily $\PiAcal(S)\Hcal\cap\PiBcal(T)\Hcal\not=\{0\}$, sufficiently coarse-grained measurements can be compatible.

When $\Acal,\Bcal$ are COINC, $\langle a_i|b_j\rangle\not=0$ for all $i,j$ and hence they are STROINC. Indeed, if for example $\langle a_1|b_1\rangle=0$, then $|a_1\rangle$ belongs to $\PiAcal(S)\Hcal\cap \PiBcal(T)\Hcal$ for $S=\{1\}$ and $T=\{2,\dots, d\}$. But since $|T|+|S|=d$, this contradicts the definition. 
Hence each basis vector $|a_i\rangle$ of $\Acal$   has full $\Bcal$-support and $\mab>0$. 
Suppose one now introduces some uncertainty in the pre-measurement state $\psi$ by considering a coherent superposition of two basis vectors $\psi=c_1|a_1\rangle + c_2|a_2\rangle$ with $c_1\not=0\not=c_2$. Then, by choosing $c_1,c_2$ to ensure $\langle b_1|\psi\rangle=0$ one reduces the uncertainty on the measurement outcomes of a measurement in the basis $\Bcal$ in the sense that the post-measurement state can no longer be  $|b_1\rangle$. In other words, we can give up some of the precision on a measurement in $\Acal$ in order to reduce the uncertainty on a measurement in $\Bcal$.   Provided $\Acal$ and $\Bcal$ are COINC, it then follows that $\langle b_j|\psi\rangle\not=0$, for all $j\not=1$. Indeed, if for example $\langle b_2|\psi\rangle=0$ as well, then $\psi\in \PiAcal(S)\Hcal\cap \PiBcal(T)\Hcal$ with $S=\{1,2\}$ and $T=\{3,4,\dots, d\}$ so that $\nab(\psi)\leq d$, which is a contradiction.   In this sense, when the bases are COINC, the increase of information on a measurement in $\Bcal$ is constrained optimally by the loss of information on a measurement in $\Acal$. 

\subsection{Complete incompatibility: a criterion and examples}
A useful criterion for  complete incompatibility
 is:
\begin{lemma}\label{lem:minorcrit}
$\Acal$ and $\Bcal$ are COINC if and only if none of the minors of the matrix $U$  vanishes. 
\end{lemma}
Recall that a $k$-minor of $U$ is the determinant of a $k$ by $k$ submatrix of $U$ obtained by removing $d-k$ rows and $d-k$ columns from $U$. The statement and proof are implicit in~\cite{Tao05}; we give a straightforward argument using linear algebra in Appendix~\ref{s:minorcrit}.

As an immediate application, one sees that, in dimensions $d=2$ and $d=3$, two bases $\Acal$ and $\Bcal$ are COINC  iff they are STROINC, \emph{i.e.} iff $1\leq i,j\leq d$, $\langle a_i|b_j\rangle\not=0$.  In dimension $2$ this is obvious. In dimension $3$, note that each column of $U$ is a multiple of the complex conjugate of the vector product of the two other columns. Since the components of the vector product are minors of order $2$,  their nonvanishing follows from the nonvanishing of all matrix elements of $U$. Hence the bases are COINC. 
In dimension more than three, the above is no longer true. Indeed, it is proven in~\cite{Tao05} that none of the minors of the DFT transition matrix  vanish if and only if the dimension $d$ is a prime number.  Lemma~\ref{lem:minorcrit} then implies the DFT is COINC iff $d$ is a prime number. This implies MUBs -- although maximally incompatible in the sense explained above -- are not necessarily COINC. In fact, when $d=4$, no MUBs are COINC, as easily seen from their explicit expression in Eqn.~\eqref{eq:MUB4}, and using Lemma~\ref{lem:minorcrit}. 

\subsection{Complete incompatibility, support uncertainty, and KD-nonclassicality.}
The link between complete incompatibility and support uncertainty is given by:
\begin{theorem} \label{thm:COINC_UNC}
$\Acal$ and $\Bcal$ are COINC iff $\nabmin=d+1$.
\end{theorem}
\noindent\emph{Proof.} For all states $\psi$,  $\psi\in \PiAcal(S_\psi)\Hcal\cap \PiBcal(T_\psi)\Hcal\not=\{0\}$. If the bases are COINC, this implies $|S_\psi|+|T_\psi|>d$. so $\nabmin>d$. Since for the basis vectors we know $\nab(|a_i\rangle)=d+1$, we conclude $\nabmin=d+1$. 
We prove the converse by proving its contraposition.  Suppose $\Acal$ and $\Bcal$ are not COINC. Then there exist $S,T$, with $|S|+|T|\leq d$ and  $\Pi_\Acal(S)\Hcal\cap\Pi_\Bcal(T)\Hcal\not=\{0\}$. Let $0\not=\psi\in\Pi_\Acal(S)\Hcal\cap\Pi_\Bcal(T)\Hcal$. For this state $n_{\Acal}(\psi)\leq |S|, n_{\Bcal}(\psi)\leq |T|$. Hence $n_{\Acal, \Bcal}(\psi)\leq d$ and $\nabmin\leq d$. \qed\\
The theorem asserts that the uncertainty diagram of $\Acal, \Bcal$ lies above the KD-nonclassicality edge iff the bases are COINC. This is illustrated in Fig.~\ref{fig:uncdiag} where panels (a), (d) and (e) show the uncertainty diagram of bases that are not COINC as follows from our previous analysis: one does indeed observe states below the nonclassicality edge. Panels (b) and (c) on the other hand concern COINCs.

When two bases are COINC, Theorems~\ref{thm:NCbound_tris} and~\ref{thm:COINC_UNC} imply that all KD-classical states $\psi$ have minimal support uncertainty: $\nab(\psi)=\nabmin=d+1$ (Fig.~\ref{fig:uncdiag}-(b)\&(c)). Note however that it is not true that all states with minimal support uncertainty are KD-classical contrary to what happens with conjugate continuous variables $Q$ and $P$.  In other words, when the bases are COINC all states, except possibly some of those with minimal support uncertainty are KD-nonclassical.


\section{Conclusion}
It has been observed recently that KD-nonclassical states can furnish a quantum advantage. This raises the question under what conditions on the observables used to define the KD-distribution such KD-nonclassicality prevails in Hilbert space? We have established that complete incompatibility -- a notion we introduce -- implies only states with minimal support uncertainty can be KD-classical, all others being KD-nonclassical. This is therefore the optimal situation when the goal is to reduce to a minimum the presence of KD-classical states.  Our findings further imply that if two bases are mutually unbiased as well as completely incompatible, their incompatibility  mimicks most closely the incompatibility of two conjugate observables in continuous variable quantum mechanics.
When the bases are strongly but not completely incompatible, the support uncertainty of states still can serve as a witness of KD-nonclassicality. A number of open questions will be explored elsewhere~\cite{SDB21}, in particular the link between strong incompatiblity and noncommutativity.
Our findings provide an improved understanding of the general structural properties of KD-nonclassicality which we expect to be important for the conception of  experiments and protocols capable of harnessing KD-nonclassicality for quantum information tasks. 
\acknowledgments
This work was supported in part by the Agence Nationale de la Recherche under grant ANR-11-LABX-0007-01 (Labex CEMPI) and by the Nord-Pas de Calais Regional Council and the European Regional Development Fund through the Contrat de Projets \'Etat-R\'egion (CPER). The author thanks David Arvidsson-Shukur for stimulating discussions on the subject matter of this work. 

\vglue0.5cm
\appendix

\centerline{\bf APPENDICES}

\section{Proof of Theorem~\ref{thm:NCbound_tris}} \label{s:proofthm}
The proof follows from a refinement of the arguments in~\cite{ArChHa20}. We proceed by contradiction and suppose $\psi$ is KD-classical. Since the KD-distribution is insensitive to global phase rotations $|a_i\rangle \to \exp(i\phi_i)|a_i\rangle, |b_j\rangle \to \exp(i\phi_j')|b_j\rangle$, we can suppose that all $\langle a_i|\psi\rangle$ and $\langle \psi|b_j\rangle$ are nonnegative (hence real) for $1\leq i, j\leq  d$. Possibly relabeling the basis vectors, we can suppose that $\langle a_i|\psi\rangle\not=0\not=\langle b_j|\psi\rangle $ for $1\leq i\leq  n_\Acal(\psi), 1\leq j\leq n_\Bcal(\psi)$ whereas all other $\langle a_i|\psi\rangle$, $\langle b_j|\psi\rangle$ vanish.  By hypothesis, the KD-distribution of $\psi$ is real and nonnegative. Hence, for the same range of $i$ and $j$, we can  conclude $\langle a_i|b_j\rangle$ is real and nonnegative. Since, by hypothesis, $\langle a_i|b_j\rangle\not=0$, this implies that $\langle a_i|b_j\rangle>0$ for these values of $i$ and $j$. 
 
Suppose now first that $n_\Acal(\psi)=d=n_\Bcal(\psi)$. Then all matrix elements $U_{ij}$ are real and nonnegative. This is in contradiction with the fact that the columns of $U$ are orthogonal. Let us now assume that $n_\Bcal(\psi)\leq n_\Acal(\psi)<d$. 
Then, for $1\leq j< j'\leq n_\Bcal(\psi)$, we have
$$
0=\langle b_j|b_{j'}\rangle=\sum_{i=1}^{n_\Acal(\psi)} \langle b_j|a_i\rangle \langle a_i|b_{j'}\rangle + \sum_{i=n_\Acal(\psi)+1}^{d} \langle b_j|a_i\rangle \langle a_i|b_{j'}\rangle.
$$
From the above, we know that 
\begin{equation}\label{eq:inprod}
\sum_{i=1}^{n_\Acal(\psi)} \langle b_j|a_i\rangle \langle a_i|b_{j'}\rangle>0.
\end{equation}

It then follows that, for all $1\leq j< j'\leq n_\Bcal(\psi)$, one has
$$
\sum_{i=n_\Acal(\psi)+1}^{d} \langle b_j|a_i\rangle \langle a_i|b_{j'}\rangle<0.
$$
Defining, for each $1\leq j\leq n_\Bcal(\psi)$ the vector $d_j=(\langle a_{n_\Acal(\psi)+1}|b_{j}\rangle,\dots, \langle a_{d}|b_{j}\rangle)\in\C^{d-n_\Acal}$ we see from the above that $\langle d_j|d_{j'}\rangle<0$. It then follows from Lemma~\ref{lem:obtuse} below that 
$$
n_\Bcal(\psi)\leq d-n_\Acal(\psi)+1.
$$
This proves the result. 
The case where $n_\Acal(\psi)\leq n_\Bcal(\psi)<d$ is treated similarly, inverting the roles of the columns and the rows. \qed

It remains to prove the following lemma, which puts an upper bound on the number of vectors in $\C^n$ that can have an obtuse angle between them, two by two.  It is a refinement of an argument in~\cite{ArChHa20}. 
\begin{lemma}\label{lem:obtuse}
Let $n,k\in\N_*$ and $v_1, v_2, \dots, v_k\in \C^n$. Then the following holds:
if $\langle v_i|v_j\rangle< 0$ for all $1\leq i<j\leq k$, then $k\leq n+1$.
\end{lemma}
\begin{proof} 
The proof goes by induction. For $n=1$, one may note that one can always take $v_1>0$, by applying a common phase rotation to all $v_i\in\C$, which does not change the inner products $\overline v_i v_j$  between them. Hence $v_j<0$ for all $j\not=1$. But if $k>2$, then this contradicts the requirement that $v_2v_3>0$. So $k\leq 2$ when $n=1$. Suppose now the result holds for some $n\in\N_*$. We show it holds for $n+1$. Let $v_1, \dots, v_k\in\C^{n+1}$. As above, we can suppose $v_1=a_1e_1$, $a_1>0$. Write $v_j=a_je_1 +w_j$, with $\langle w_j, e_1\rangle=0$, for all $j=2,\dots, k$. By hypothesis, $\langle v_1|v_j\rangle<0$, so that all $a_j< 0$. As a result, for all $2\leq i<j\leq k$,
$$
0> \langle v_i|v_j\rangle=a_ia_j +\langle w_i|w_j\rangle.
$$
Hence, for all $2\leq i<j\leq k$, $\langle w_i|w_j\rangle< 0$. 
They therefore constitute $(k-1)$ nonvanishing vectors in $\C^n$, and since their mutual inner products are all negative, the induction hypothesis allows to conclude that $k-1\leq n+1$ so that $k\leq (n+1)+1$ which is the desired result. 
\end{proof}

\section{Uncertainty diagrams for the mutually unbiased bases  in dimension $4$}\label{s:MUB4}
We will, in this section, draw up the uncertainty diagrams for all mutually unbiased bases (MUBs) in dimension $d=4$. Before doing so, we remark that the definitions of strong and complete incompatibility as well as of mutual unbiasedness depend on the two bases $\Acal$ and $\Bcal$ only through the unitary transition operator between them, defined as 
$U|a_j\rangle=|b_j\rangle$, with matrix elements $U_{ij}=\langle a_i|b_j\rangle$ in the $\Acal$-basis. Indeed, if $\Acal'$ and $\Bcal'$ are two other bases, where, for some unitary operator $V$, $|a'_i\rangle=V|a_i\rangle$ and $|b'_j\rangle=V|b_j\rangle$, then $\Acal$ and $\Bcal$ are  strongly/completely incompatible or mutually unbiased if and only if 
$\Acal'$ and $\Bcal'$ are, as is readily checked; the unitary transition matrix is the same in both cases. 
One could in fact identify $\Hcal$ with $\C^d$ and systematically use the canonical basis of $\C^d$ as the $\Acal$-basis. The choice of the $\Bcal$ basis is then completely determined by the choice of a unitary transition matrix $U$. 
We will say $U$ is strongly/completely incompatible or mutually unbiased whenever the bases $\Acal$ and $\Bcal$ are. These properties are also not affected  by a renumbering of the basis vectors, nor by global phase changes of the basis vectors as is also readily checked. 

Up to permutations of rows and columns, and global phase rotations, the transition matrices of mutually unbiased bases (MUBs) in dimension $4$ are known to be all of the form~\cite{Ba19}
\begin{equation}\label{eq:MUB4}
U(s)=\frac12
\begin{pmatrix}
1&1&1&1\\
1&1&-1&-1\\
1&-1&s&-s\\
1&-1&-s&s
\end{pmatrix}, \quad |s|=1; U^\dagger(s)=U(\overline s).
\end{equation}
The case $s=i$ corresponds to the discrete Fourier transform (DFT) in dimension $4$. It is clear $U$ is not completely incompatible (COINC) since it has 
vanishing $2$-minors (Lemma~\ref{lem:minorcrit}). We now draw up the uncertainty diagram of $U(s)$, for $s\not=\pm 1$. The results are shown in Fig.~\ref{fig:uncdiag}-(a).

It  is readily checked that for all points $(\na,\nb)$  above the nonclassicality edge $\na+\nb=5$  there exists $\psi\in\Hcal$ so that $\na(\psi)=\na$ and $\nb(\psi)=\nb$.   Since Theorem~\ref{thm:NCbound_tris} applies we can conclude that all corresponding states are KD-nonclassical. 

Consider then states along the nonclassicality edge $n_\Acal(\psi)+n_\Bcal(\psi)=5$.
The basis vectors, with $n_\Acal(\psi)=1, n_\Bcal(\psi)=4$ or \emph{vice versa},  are in this case; they are KD-classical. 

Next, if $\psi$ is such that $T_\psi=\{1,2\}$, then $|S_\psi|\not=3$ as is readily checked.  The same is true when $T_\psi=\{3,4\}$. 

Consider therefore $\psi$ so that $T_\psi=\{1,3\}$. To obtain $|S_\psi|=3$, one needs $\psi=\psi_\pm$, where
\begin{eqnarray*}
\psi_+&=&\frac1{\sqrt2}(|b_1\rangle + |b_3\rangle)\\
&=&\frac1{2\sqrt2}\left(2|a_1\rangle +(1+s)|a_3\rangle +(1-s)|a_4\rangle\right)\\
\psi_-&=&\frac1{\sqrt2}(|b_1\rangle - |b_3\rangle)\\
&=&\frac1{2\sqrt2}\left(2|a_2\rangle +(1-s)|a_3\rangle +(1+s)|a_4\rangle\right).
\end{eqnarray*}
Then $S_{\psi_+}=\{1,3,4\}$ and $S_{\psi_-}=\{2,3,4\}$ since $s\not=\pm1$. The corresponding KD-distributions are 
\begin{eqnarray*}
Q_+&=&\frac14
\begin{pmatrix}
1&0&1&0\\
0&0&0&0\\
1/2(1+\overline s)&0&1/2(1+s)&0\\
1/2(1-\overline s)&0&1/2(1-s)&0
\end{pmatrix},\\
Q_-&=&\frac14
\begin{pmatrix}
0&0&0&0\\
1&0&1&0\\
1/2(1-\overline s)&0&1/2(1-s)&0\\
1/2(1+\overline s)&0&1/2(1+s)&0
\end{pmatrix}.
\end{eqnarray*}
Consequently, $\psi_\pm$ are both KD-nonclassical. The case where $T_\psi=\{1,4\}$ is similar, with $s$ replaced by $-s$. 
The cases where $T_\psi=\{2,3\}$ or $T=\{2,4\}$ and $|S_\psi|=3$ are obtained by permuting the rows and/or the columns of the KD-distributions $Q_\pm$ obtained above. One therefore sees in this example that, along the nonclassicality edge, there co-exist both KD-classical states and KD-nonclassical states. The only KD-classical states are the basis vectors. Since $U^\dagger(s)=U(\overline s)$, the same results are obtained when one switches the roles of $\Acal$ and $\Bcal$. 

As a result of Eq.~\eqref{eq:uncprincab} it only remains to investigate the states for which $n_{\Acal,\Bcal}(\psi)=4$. This can only happen when $n_\Acal(\psi)=2=n_\Bcal(\psi)$ since $U$ has no zeros. 
Straightforward calculations show such states exist and that they are all KD-classical. The latter statement also follows from~\eqref{eq:NCupperbound}, which  can be rewritten as 
\begin{equation}\label{eq:NCupperboundbis}
1\leq \KDNC(\psi)\leq \Mab\sqrt{n_\Acal(\psi)n_\Bcal(\psi)},
\end{equation}
where in the present case $\Mab=\frac12$ 
\begin{equation}\label{eq:KDNC}
\KDNC(\psi)=\sum_{i,j}|Q_{ij}|\geq 1,
\end{equation}
is the so-called nonclassicality of the state. Clearly, a state is KD-nonclassical if and only if $\KDNC(\psi)>1$.
Hence Eqn.~\eqref{eq:NCupperboundbis} implies that if $\psi$ saturates the uncertainty principle~\eqref{eq:uncprincab}, then it is KD-classical.

When $s=\pm 1$, the situation is slightly different. There are then, for example, no states with $|S_\psi|=3$ and $|T_\psi|=2$, as is readily checked. 

In Fig.~\ref{fig:uncdiag}-(b) we show the uncertainty diagram of the unitary $U(\epsilon)$ of Eqn.~\eqref{eq:Upert}, with $\epsilon=0.1$, which is a perturbation of a MUB. 
We checked numerically that all its minors are nonvanishing, thereby proving it is COINC.  When $\epsilon\not=0$, it is no longer possible to obtain simple closed analytical expressions for the various states along the nonclassicality edge. We therefore proceeded to a numerical computation of those states and we established they are KD-nonclassical, with the exception of the basis states themselves. 


\section{The discrete Fourier transform}\label{s:DFT}
Recall that the discrete Fourier transform (DFT) is 
$$
U_{ij}=\langle a_i|b_j\rangle= \frac{1}{\sqrt{d}}\exp(i\frac{2\pi}{d}ij).
$$
It is convenient to choose $0\leq i,j\leq d-1$. When $d$ is prime, it follows from Lemma~\ref{lem:minorcrit} and~\cite{Tao05} that the DFT is COINC. As a result, $n_\Acal(\psi)+n_\Bcal(\psi)\geq d+1$ for all $\psi$. This is no longer true when $d$ is not prime: in that case, there do exist states for which $n_\Acal(\psi)+n_\Bcal(\psi)< d+1$ as we will see below. This is illustrated in Fig.~\ref{fig:uncdiag}-(c)\&(d).

For all dimensions, Theorem~\ref{thm:NCbound_tris} implies that all states for which $n_\Acal(\psi)+n_\Bcal(\psi)>d+1$ are KD-nonclassical. It remains therefore to  investigate the KD-nonclassicality of the states with $n_\Acal(\psi)+n_\Bcal(\psi)\leq d+1$. Recall that the support uncertainty principle guarantees that
\begin{equation*}\label{eq:uncprincMUB}
n_\Acal(\psi)n_\Bcal(\psi)\geq d.
\end{equation*}
When $d$ is not prime, there may therefore exist states in the region of the $(n_\Acal, n_\Bcal)$-plane on or above the hyperbola $n_\Acal n_\Bcal= d$ and on or below the straight line segment $n_{\Acal}+ n_\Bcal=d+1$ that we refer to as the \emph{nonclassicality edge}. We know that all KD-classical states lie in this region. But, as we shall illustrate, there may also be some KD-nonclassical states there. We will investigate these issues now, one dimension at a time.

When $d=2$, and $U$ is the DFT, it follows from what precedes that all states except the basis states are KD-nonclassical.

When $d=3$, the DFT is equivalent, after dephasing and permuting rows and columns, to 
\begin{equation}\label{eq:DFT3}
U=\frac1{\sqrt3}
\begin{pmatrix}
1&1&1\\
1&\omega&\omega^2\\
1&\omega^2&\omega
\end{pmatrix},
\quad \textrm{where}\quad \omega=\exp(i\frac{2\pi}{3}).
\end{equation}
Using Lemma~\ref{lem:minorcrit} it is easily checked to be COINC, a fact that also follows also from the general result cited, since $3$ is prime. Straigthforward explicit computations show that, when $|S_\psi|=2=|T_\psi|$, $\psi$ is KD-nonclassical. 
As in dimension $2$, the only classical states are the basis vectors.

The case $d=4$  was treated in the previous section ($s=i$). Note that the classical states are now on the hyperbola $n_\Acal n_\Bcal=d$ which contains one point $n_\Acal(\psi)=2=n_\Bcal(\psi)$ strictly below the nonclassicality edge $n_\Acal(\psi)+n_\Bcal(\psi)=d+1=5$. So this is the lowest dimension for which the DFT admits classical states other than the basis states (Fig.~\ref{fig:uncdiag}-(a)).

When $d=5$ the DFT is COINC, since $5$ is a prime number. There are now no states below the KD-nonclassicality edge $n_\Acal+n_\Bcal=6$. It remains to investigate the KD-nonclassicality for the states on the nonclassicality edge. Apart from the basis states, which are classical, they are of three types: $n_\Acal(\psi)=2$, $n_\Bcal(\psi)=4$; $n_\Acal(\psi)=3$, $n_\Bcal(\psi)=3$; $n_\Acal(\psi)=4$, $n_\Bcal(\psi)=2$. We computed those states numerically and showed  they are nonclassical by observing numerically that $\KDNC(\psi)>1$. (See Fig.~\ref{fig:uncdiag}-(c).)

 Suppose now $d$ is not prime and  $2< d=pq$, with $1<p,q<d$; the divisors $p$ and $q$ need not be prime.   One can then show that that the DFT is not COINC by noting that there exist states $\psi$ for which $n_\Acal(\psi)+n_\Bcal(\psi)\leq d$. In fact, it is easy to construct states for which $n_\Acal(\psi)=q, n_\Bcal(\psi)=p$, so that the lower bound $n_\Acal(\psi)n_\Bcal(\psi)=d$ is reached in this case.  
For that purpose, consider, for $0\leq m<p, 0\leq s<q$, 
 $$
 |m,s\rangle =\frac1{\sqrt{q}} \sum_{k=0}^{q-1} \exp(i\frac{2\pi}{q}sk) |a_{kp+m}\rangle,
 $$
 which are readily checked to form an orthonormal basis. 
Clearly $n_\Acal(|m,s\rangle)= q$. For $j=\ell q+r$, $0\leq \ell<p$, $0\leq r<q$, one has
$$
\langle b_j| m,s\rangle=\frac1{\sqrt{q}}\exp(-i\frac{2\pi}{d}mj)\delta _{sr}.
$$
Hence $n_\Bcal(|m,s\rangle)=p$. It can be shown (see~\cite{MaOzPr04} and references therein) that the $|m,s\rangle$ are the only states with the property that $n_\Acal(\psi)n_\Bcal(\psi)=d$. %
%
As a result of~\eqref{eq:NCupperboundbis}, these states are all KD-classical.
The smallest non-prime dimension is $d=4$, which we already treated above.
When $d=6$, the previous development shows there are classical states with $n_\Acal(\psi)=2, n_\Bcal(\psi)=3$ and \emph{vice versa}. We know from Theorem~\ref{thm:NCbound_tris} that all states with $n_\Acal(\psi)+n_\Bcal(\psi)> 7$ are KD-nonclassical. However, we are no longer ensured that, given a point $(n_\Acal, n_\Bcal)$ with $n_\Acal+n_\Bcal>7$, there exists a state with $n_\Acal(\psi)=n_\Acal, n_\Bcal(\psi)=n_\Bcal$ since the DFT is not COINC when $d=6$. In addition, there may now be points $(n_\Acal, n_\Bcal)$ below the nonclassicality edge and above $n_\Acal n_\Bcal=d$ for which such a state does exist. 

We now explore these phenomena. 

We first consider all states for which $n_\Bcal(\psi)=2$. It is easy to check that, for $0\leq k_1<k_2\leq 5$, $0\leq i_1<5$ 
$$
|k_1, k_2, i_1\rangle=\frac1{\sqrt2}\left(\omega^{i_1k_2}|b_{k_1}\rangle-\omega^{i_1k_1}|b_{k_2}\rangle\right)
$$
is the only state with the property that 
$$
\langle a_{i_1}|k_1, k_2, i_1\rangle=0.
$$
Hence $n_\Acal(k_1,k_2,i_1)\leq 5$. Now note that, with $0\leq i_1<i_2\leq 5$, 
$$
\langle a_{i_2}|k_1, k_2, i_1\rangle=0
$$
as well if and only if $(k_2-k_1)(i_2-i_1)=0$ modulo $6$. Hence, if this condition is not satisfied, then $n_{\Acal, \Bcal}(k_1, k_2, i_1)=2+5=7$ and if it is,
$n_{\Acal, \Bcal}(k_1, k_2, i_1)\leq 2+4=6$.  We have computed the nonclassicality of these states numerically and have observed the states are KD-nonclassical except when $n_{\Acal, \Bcal}(k_1, k_2, i_1)=2+3=5$, as we saw above. The same observations hold true when the roles of $\Acal$ and $\Bcal$ are interchanged. 

We have further established through a numerical computation that there are no states for which $n_\Acal=3=n_\Bcal$. For that purpose, we constructed all states
of the form
$$
|k_1, k_2, k_3\rangle =d_{k_1}|b_{k_1}\rangle+d_{k_2}|b_{k_2}
\rangle+d_{k_3}|d_{k_3}\rangle,
$$ 
for which $\langle a_{i_1}|k_1,k_2,k_3\rangle=0=\langle a_{i_2}|k_1, k_2, k_3\rangle=\langle a_{i_3}|k_1, k_2, k_3\rangle$,
with $0\leq k_1<k_2<k_3\leq 5$, $0<i_1<i_2<i_3\leq 5$ and computed for each $n_\Acal$ and  $n_\Bcal$. Clearly, for all such states,
$n_\Acal\leq 3, n_\Bcal\leq 3$. As it turns out, equality is achieved for none. 

We similarly explored the pairs  $(n_\Acal, n_\Bcal)=(3,4)$ and $(4,3)$.  We found numerically all $\psi$ so that $n_\Acal(\psi)=3$, $n_\Bcal(\psi)=4$ and observed they are all KD-nonclassical. The results of these computations are summarized in Fig.~\ref{fig:uncdiag}-(d).

It is of course impractical to treat the case of general $d$ in this manner. An analytic treatment would probably involve a fair amount of number theory, such as the techniques used to prove that for prime $d$, the DFT is COINC~\cite{Tao05}. A question that comes to mind when looking at the uncertainty diagrams for the dimensions $2\leq d\leq 6$ is whether it is true that the only KD-classical states for the DFT are the ones on the hyperbola $\na \nb =d$? For prime dimensions, this would mean the only KD-classical states are the basis states. 

\section{Spin $s$}\label{s:spins}
For general integral spin $s$ we show that the $J_x$ and $J_z$ bases are not STROINC. For that purpose, we need to investigate the matrix elements
\begin{eqnarray*}
\langle J_z=m'|J_x=m\rangle&=&\langle J_z=m'|\exp(-i\beta J_y|J_z=m\rangle\\
&=&d_{m',m}^{(s)}(\beta)
\end{eqnarray*}
with $\beta=\frac{\pi}{2}$ and where the matrix $d_{m',m}^{(s)}(\beta)$ is Wigner's well-known little matrix~\cite{SA94}. It is known that
$$
d_{m,m'}^{(s)}(\beta)=(-1)^{m'-m}d_{m',m}^{(s)}(\beta).
$$
From Wigner's formula for $d_{m',m}^{(s)}(\pi/2)$ one readily concludes that, for all integer $s$, and $m\geq 0$,
$$
d_{0,m}^{(s)}(\pi/2)=s!\sum_{k=m}^s(-1)^k \frac{\sqrt{(s+m)!(s-m)!}}{(s-k)k! (s-(k-m))!(k-m)!}.
$$
Making the change of variables $k'=s+m-k$, one observes that $d_{0,m}^{(s)}(\pi/2)=0$ for all odd $m$ when $s$ is even and for all even $m$ when $s$ is odd. Indeed, there are $s-m+1$ terms in the sum above and they cancel two by two whenever this number is even, which means $s-m$ must be odd. 

When $s=2$, the explicit form of the Wigner matrix is readily computed:
$$
U=\frac12
\begin{pmatrix}
\frac12&1&\sqrt{\frac{3}{2}}&1&\frac12\\
-1&-1&0&1&1\\
\sqrt{\frac{3}{2}}&0&-1&0&\sqrt{\frac{3}{2}}\\
-1&1&0&-1&1\\
\frac12&-1&\sqrt{\frac{3}{2}}&-1&\frac12
\end{pmatrix}.
$$
We constructed the uncertainty diagram of $U$ numerically, proceeding as for the DFT in the previous section. The result is displayed in Fig.~\ref{fig:uncdiag}-(e). Note that, along the nonclassicality
edge, when $\na=2, \nb=4$ (or vice versa), there exist both KD-classical and KD-nonclassical states in this case.


\section{Proof of Lemma~\ref{lem:minorcrit}}\label{s:minorcrit}
\emph{Proof of} $\Rightarrow$. We proceed by contraposition. Suppose a $k$-minor of $U$ vanishes,  for some $1\leq k\leq d-1$. Reordering the rows and columns of $U$ (which amounts to  relabeling the basis vectors), we can assume that the matrix 
$$
\begin{pmatrix}
\langle a_1|b_1\rangle&\dots&\langle a_1|b_k\rangle\\
\vdots&\vdots&\vdots\\
\langle a_k|b_1\rangle&\dots&\langle a_k|b_k\rangle
\end{pmatrix}
$$
has vanishing determinant. Hence there exist $(\beta_1,\dots\beta_k)\not=0$ so that 
$
|\psi\rangle =\sum_{j=1}^k \beta_j |b_j\rangle
$
satisfies, for all $i\in \llbracket 1,k\rrbracket$, $\langle a_i|\psi\rangle=0$. So $\psi\in \PiAcal(S)\Hcal\cap \PiBcal(T)$ with $S=\llbracket k+1,d\rrbracket$ and $T=\llbracket 1,k\rrbracket$. Since $|S|+|T|=d$, this implies $\Acal$ and $\Bcal$ are not COINC. \\
\emph{Proof of} $\Leftarrow$. We proceed again by contraposition. Suppose $\Acal$ and $\Bcal$ are not COINC. Then there exist $S, T\subset \llbracket 1, d\rrbracket$ with $|S|+|T|\leq d$, and for which $\PiAcal(S)\Hcal\cap \PiBcal(T)\Hcal\not=\{0\}$.  Let $0\not=\psi\in \PiAcal(S)\Hcal\cap \PiBcal(T)\Hcal$. Possibly reordering the elements of $\Bcal$, we can suppose $T=\llbracket 1, k\rrbracket$ and write $\psi=\sum_{j=1}^k\beta_j|b_j\rangle$. By hypothesis $\PiAcal(S^c)|\psi\rangle=0$ so that, for all $i\in S^c$, $\langle a_i|\psi\rangle=0$. Since $|S|+|T|\leq d$, $|T|\leq |S^c|$. Reordering the basis $\Acal$, we can assume $S^c=\llbracket 1, \ell\rrbracket$, with $k\leq \ell$. It follows that 
$$
\begin{pmatrix}
\langle a_1|b_1\rangle&\dots&\langle a_1|b_k\rangle\\
\vdots&\vdots&\vdots\\
\langle a_\ell|b_1\rangle&\dots&\langle a_\ell|b_k\rangle
\end{pmatrix}
\begin{pmatrix}
\beta_1\\ \vdots \\ \beta_k
\end{pmatrix}
=0.
$$
Since $(\beta_1,\dots\beta_k)\not=0$, one can extract from the $\ell$ by $k$  matrix in the left-hand side a vanishing $k$-minor of $U$. \qed


%


\bibliographystyle{apsrev4-1}

%

\end{document}